\documentclass[a4paper,UKenglish,cleveref, autoref]{lipics-v2019}

\usepackage[ruled]{algorithm}
\usepackage[noend]{algpseudocode}
\usepackage{thmtools, thm-restate}

\title{Cartesian Tree Matching and Indexing}

\author{Sung Gwan Park}{Seoul National University, Korea}{sgpark@theory.snu.ac.kr}{}{}

\author{Amihood Amir}{Bar-Ilan University, Israel}{amir@esc.biu.ac.il}{}{}

\author{Gad M. Landau}{University of Haifa, Israel and New York University, USA}{landau@univ.haifa.ac.il}{}{}

\author{Kunsoo Park\footnote{Corresponding author}}{Seoul National University, Korea}{kpark@theory.snu.ac.kr}{}{}

\authorrunning{S.G. Park, A. Amir, G.M. Landau, and K. Park}

\Copyright{Sung Gwan Park, Amihood Amir, Gad M. Landau, and Kunsoo Park}



\nolinenumbers

\EventEditors{Nadia Pisanti and Solon P. Pissis}
\EventNoEds{2}
\EventLongTitle{30th Annual Symposium on Combinatorial Pattern Matching (CPM 2019)}
\EventShortTitle{CPM 2019}
\EventAcronym{CPM}
\EventYear{2019}
\EventDate{June 18--20, 2019}
\EventLocation{Pisa, Italy}
\EventLogo{}
\SeriesVolume{128}
\ArticleNo{13}

\begin{document}
\maketitle
\begin{abstract}
We introduce a new metric of match, called \emph{Cartesian tree matching}, which means that two strings match if they have the same Cartesian trees. 
Based on Cartesian tree matching, we define single pattern matching for a text of length $n$ and a pattern of length $m$, and multiple pattern matching for a text of length $n$ and $k$ patterns of total length $m$.
We present an $O(n+m)$ time algorithm for single pattern matching, and an $O((n+m) \log k)$ deterministic time or $O(n+m)$ randomized time algorithm for multiple pattern matching.
We also define an index data structure called Cartesian suffix tree, and present an $O(n)$ randomized time algorithm to build the Cartesian suffix tree.
Our efficient algorithms for Cartesian tree matching use a representation of the Cartesian tree, called the \emph{parent-distance representation}. 
\end{abstract}

\section{Introduction}

String matching is one of fundamental problems in computer science, and it can be applied to many practical problems. In many applications string matching has variants derived from exact matching (which can be collectively called \emph{generalized matching}), such as order-preserving matching \cite{Order-preservingMatching:Ternary, Order-preservingMatching, Order-preservingMatching2}, parameterized matching \cite{ParameterizedMatching2, ParameterizedMatching, SuffixTree:Parameterized}, jumbled matching \cite{JumbledMatching}, overlap matching \cite{OverlapMatching}, pattern matching with swaps \cite{SwappedMatching}, and so on. These problems are characterized by the way of defining a \emph{match}, which depends on the application domains of the problems. In financial markets, for example, people want to find some patterns in the time series data of stock prices. In this case, they would like to know more about some pattern of price fluctuations than exact prices themselves \cite{StockTimeSeries}. Therefore, we need a definition of match which is appropriate to handle such cases.

\begin{figure}[b]
\centering
    \includegraphics[height=4.5cm]{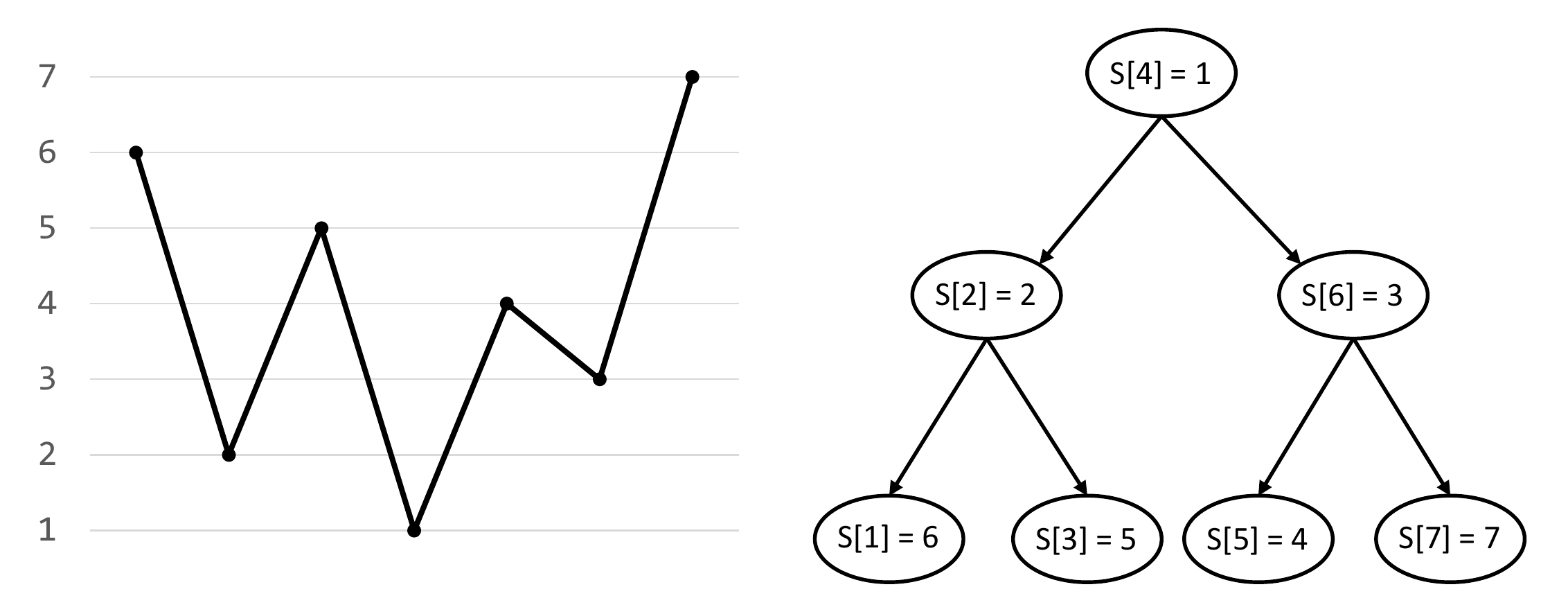}
    \caption{Example pattern $S=(6, 2, 5, 1, 4, 3, 7)$ and its corresponding Cartesian tree}
    \label{fig:cartesiantreematching}
\end{figure}

The Cartesian tree \cite{CartesianTree:Definition} is a tree data structure that represents an array, only focusing on the results of comparisons between numeric values in the array. In this paper we introduce a new metric of match, called \emph{Cartesian tree matching}, which means that two strings match if they have the same Cartesian trees. If we model the time series stock prices as a numerical string, we can find a desired pattern from the data by solving a Cartesian tree matching problem. For example, let's assume that the pattern we want to find looks like the picture on the left of Figure \ref{fig:cartesiantreematching}, which is a common pattern called the head-and-shoulder \cite{StockTimeSeries} (in fact there are two versions of the head-and-shoulder: one is the picture in Figure \ref{fig:cartesiantreematching} and the other is the picture reversed). 
The picture on the right of Figure \ref{fig:cartesiantreematching} is the Cartesian tree corresponding to the pattern on the left.
Cartesian tree matching finds every position of the text which has the same Cartesian tree as the picture on the right of Figure \ref{fig:cartesiantreematching}. 

Even though order-preserving matching \cite{Order-preservingMatching:Ternary, Order-preservingMatching, Order-preservingMatching2} can also be applied to finding patterns in time series data, Cartesian tree matching may be more appropriate than order-preserving matching in finding patterns. For instance, let's assume that we are looking for the pattern in Figure \ref{fig:cartesiantreematching} in time series stock prices. An important characteristic of the pattern is that the price hit the bottom (head), and it has two shoulders before and after the head. But the relative order between the two shoulders (i.e., which one is higher) does not matter.
If we model this pattern into order-preserving matching, then order-preserving matching imposes a relative order between two shoulders $S[2]$ and $S[6]$. Moreover, it imposes an unnecessary order between two valleys $S[3]$ and $S[5]$. Hence, order preserving matching may not be able to find such a pattern in time series data.  In contrast, the pattern in Figure \ref{fig:cartesiantreematching} can be represented by one Cartesian tree, and therefore Cartesian tree matching is a more appropriate metric in such cases.

In this paper we define string matching problems based on Cartesian tree matching:
single pattern matching for a text of length $n$ and a pattern of length $m$, and multiple pattern matching for a text of length $n$ and $k$ patterns of total length $m$, 
and we present efficient algorithms for them.
We also define an index data structure called Cartesian suffix tree as in the cases of parameterized matching and order-preserving matching \cite{SuffixTree:Parameterized, Order-preservingIndexing}, and present an efficient algorithm to build the Cartesian suffix tree.
To obtain efficient algorithms for Cartesian tree matching, we define a representation of the Cartesian tree, called the \emph{parent-distance representation}. 

In Section 2 we give basic definitions for Cartesian tree matching. In Section 3 we propose an $O(n+m)$ time algorithm for single pattern matching.
In Section 4 we present an $O((n+m) \log k)$ deterministic time or $O(n+m)$ randomized time algorithm for multiple pattern matching.
In Section 5 we define the Cartesian suffix tree, and present an $O(n)$ randomized time algorithm to build the Cartesian suffix tree of a string of length $n$.

\section{Problem Definition}

\subsection{Basic notations}
A \emph{string} is a sequence of characters in an alphabet $\Sigma$, which is a set of integers. We assume that the comparison between any two characters can be done in $O(1)$ time. For a string $S$, $S[i]$ represents the $i$-th character of $S$, and $S[i..j]$ represents a substring of $S$ starting from $i$ and ending at $j$.

\subsection{Cartesian tree matching}
A string $S$ can be associated with its corresponding Cartesian tree $CT(S)$ according to the following rules \cite{CartesianTree:Definition}:
\begin{itemize}
    \item If $S$ is an empty string, $CT(S)$ is an empty tree.
    \item If $S[1..n]$ is not empty and $S[i]$ is the minimum value among $S$, $CT(S)$ is the tree with $S[i]$ as the root, $CT(S[1..i-1])$ as the left subtree, and $CT(S[i+1..n])$ as the right subtree. If there are two or more minimum values, we choose the leftmost one as the root.
\end{itemize}

\noindent
Since each character in string $S$ corresponds to a node in Cartesian tree $CT(S)$, we can treat each character as a node in the Cartesian tree.

\emph{Cartesian tree matching} is the problem to find all the matches in the text which have the same Cartesian tree as a given pattern. Formally, we define it as follows:

\begin{definition} 
\emph{(Cartesian tree matching) Given two strings text $T[1..n]$ and pattern $P[1..m]$, find every $1 \leq i \leq n-m+1$ such that $CT(T[i..i+m-1]) = CT(P[1..m])$.}
\end{definition}

For example, let's consider a sample text $T = (41, 36, 15, 8, 41, 23, 28, 16, 26, 22, 56, 29, 12, 61)$. If we find the pattern in Figure \ref{fig:cartesiantreematching}, which is $P = (6, 2, 5, 1, 4, 3, 7)$, we can find a match at position 5 of the text, i.e., $CT(T[5..11])=CT(P[1..7])$. Note that the matched text is not a match in order-preserving matching \cite{Order-preservingMatching, Order-preservingMatching2} because the relative order between $T[6]=23$ and $T[10]=22$ is different from that between $P[2]=2$ and $P[6]=3$, but it is a match in Cartesian tree matching.

\section{Single Pattern Matching in \texorpdfstring{$O(n + m)$}{Lg} Time}
\subsection{Parent-distance representation}
In order to solve Cartesian tree matching without building every possible Cartesian tree, we propose an efficient representation to store the information about Cartesian trees, called the \emph{parent-distance representation}.

\begin{definition}
    \label{def:pdr}
    \emph {(Parent-distance representation) Given a string $S[1..n]$, the} parent-distance representation \emph{of $S$ is an integer string $PD(S)[1..n]$, which is defined as follows:}
    \begin{equation*}
        PD(S)[i] = \begin{cases}
           i - \max_{1 \leq j < i}\{{j : S[j] \leq S[i]} \} &\text{\emph{if such $j$ exists}}\\
            0 &\text{\emph{otherwise}}
        \end{cases}
    \end{equation*}
\end{definition}

For example, the parent-distance representation of string $S = (2, 5, 4, 2, 2, 1)$ is $PD(S) = (0, 1, 2, 3, 1, 0)$. Note that $S[j]$ in Definition \ref{def:pdr} represents the parent of $S[i]$ in Cartesian tree $CT(S[1..i])$. Furthermore, if there is no such $j$, $S[i]$ is the root of Cartesian tree $CT(S[1..i])$.

Theorem \ref{thm:ctequivalenttopdr} shows that the parent-distance representation has a one-to-one mapping to the Cartesian tree, so it can substitute the Cartesian tree without any loss of information.

\begin{restatable}[]{theorem}{ctequivalenttopdr}
    \label{thm:ctequivalenttopdr}
    \emph{Two strings $S_1$ and $S_2$ have the same Cartesian trees if and only if $S_1$ and $S_2$ have the same parent-distance representations.}
\end{restatable}

\begin{proof}
If two strings have different lengths, they have different Cartesian trees and different parent-distance representations, so the theorem holds. Therefore, we can only consider the case where $S_1$ and $S_2$ have the same length. Let $n$ be the length of $S_1$ and $S_2$. We prove the theorem by an induction on $n$.

If $n = 1$, $S_1$ and $S_2$ will always have the same Cartesian trees with only one node. Furthermore, they will have the same parent-distance representation $(0)$. Therefore, the theorem holds when $n = 1$.

Let's assume that the theorem holds when $n = k$, and show that it holds when $n = k+1$.
 
$(\Longrightarrow)$ Assume that $S_1[1..k+1]$ and $S_2[1..k+1]$ have the same Cartesian trees (i.e., $CT(S_1[1..k+1]) = CT(S_2[1..k+1])$). There are two cases.
\begin{itemize}
    \item If $S_1[k+1]$ and $S_2[k+1]$ are not roots of the Cartesian trees, let $S_1[j]$ be the parent of $S_1[k+1]$, and $S_2[l]$ the parent of $S_2[k+1]$. 
    Since $CT(S_1[1..k+1])=CT(S_2[1..k+1])$, we have $j=l$.
    If we remove $S_1[k+1]$ from Cartesian tree $CT(S_1[1..k+1])$, we obtain the tree $CT(S_1[1..k])$, where the left subtree of $S_1[k+1]$ is attached to its parent $S_1[j]$.
    If we remove $S_2[k+1]$ from $CT(S_2[1..k+1])$, we obtain $CT(S_2[1..k])$ in the same way.
    Since $CT(S_1[1..k+1]) = CT(S_2[1..k+1])$, we get $CT(S_1[1..k]) = CT(S_2[1..k])$, and therefore $PD(S_1)[1..k] = PD(S_2)[1..k]$ by induction hypothesis.
Since $PD(S_1)[k+1]=k+1-j$ and $PD(S_2)[k+1]=k+1-l$ (and $j=l$), we have $PD(S_1) = PD(S_2)$.
\item If $S_1[k+1]$ and $S_2[k+1]$ are roots, we remove $S_1[k+1]$ and $S_2[k+1]$ to get $CT(S_1[1..k])$ and $CT(S_2[1..k])$. Since $CT(S_1[1..k+1]) = CT(S_2[1..k+1])$, we have $CT(S_1[1..k]) = CT(S_2[1..k])$, and therefore $PD(S_1)[1..k] = PD(S_2)[1..k]$ by induction hypothesis.
Since $PD(S_1)[k+1]=PD(S_2)[k+1]=0$ in this case, we get $PD(S_1) = PD(S_2)$.
\end{itemize}
    
 $(\Longleftarrow)$ Assume that $S_1[1..k+1]$ and $S_2[1..k+1]$ have the same parent-distance representations (i.e., $PD(S_1)[1..k+1] = PD(S_2)[1..k+1]$). Since $PD(S_1)[1..k] = PD(S_2)[1..k]$, we have $CT(S_1[1..k]) = CT(S_2[1..k])$ by induction hypothesis. From $CT(S_1[1..k])$, we can derive $CT(S_1[1..k+1])$ as follows.
 If $PD(S_1)[k+1]>0$, let $x$ be $S_1[k+1-PD(S_1)[k+1]]$. We insert $S_1[k+1]$ into $CT(S_1[1..k])$ so that the parent of $S_1[k+1]$ is $x$ and the original right subtree of $x$ becomes the left subtree of $S_1[k+1]$.
 If $PD(S_1)[k+1]=0$, $S_1[k+1]$ is the root of $CT(S_1[1..k+1])$ and $CT(S_1[1..k])$ becomes the left subtree of $S_1[k+1]$.
 We derive $CT(S_2[1..k+1])$ from $CT(S_2[1..k])$ in the same way.
 Since $CT(S_1[1..k]) = CT(S_2[1..k])$ and $PD(S_1)[k+1]=PD(S_2)[k+1]$, we can conclude that $CT(S_1[1..k+1]) = CT(S_2[1..k+1])$.
       
Therefore, we have proved that there is a one-to-one mapping between Cartesian trees and parent-distance representations.
\end{proof}

\subsection{Computing parent-distance representation}
Given a string $S[1..n]$, we can compute the parent-distance representation in linear time using a stack, as in \cite{Order-preservingIndexing, CartesianTree:BitRepresentation}. The main idea is that if two characters $S[i]$ and $S[j]$ for $i < j$ satisfy $S[i] > S[j]$, $S[i]$ cannot be the parent of $S[k]$ for any $k > j$. Therefore, we will only store $S[i]$ which does not have such $S[j]$ while scanning from left to right. If we store such $S[i]$ only, they form a non-decreasing subsequence of $S$. When we consider a new value, therefore, we can pop values that are larger than the new value, find its parent, and push the new value and its index into the stack. Algorithm \ref{alg:computepdref} describes the algorithm to compute $PD(S)$.

\begin{algorithm}[t]
\caption{Computing parent-distance representation of a string}\label{alg:computepdref}
\begin{algorithmic}[1]
\Procedure{PARENT-DIST-REP}{$S[1..n]$}
    \State $ST \gets \text{an empty stack}$ 
    \For{$i \gets 1$ \textbf{to} $n$}
        \While{$ST$ is not empty}
            \State $(value, index) \gets ST.top$
            \If{$value \leq S[i]$}
                \State \textbf{break}
            \EndIf
            \State $ST.pop$
        \EndWhile
        
        \If{$ST$ is empty}
            \State $PD(S)[i] \gets 0$
        \Else
            \State $PD(S)[i] \gets i - index$
        \EndIf
        
        \State $ST.push((S[i], i))$
    \EndFor
    \State \textbf{return} $PD(S)$
\EndProcedure
\end{algorithmic}
\end{algorithm}
 
Furthermore, given the parent-distance representation of string $S$, we can compute the parent-distance representation of any substring $S[i..j]$ easily. 
To compute $PD(S[i..j])[k]$, we need only check whether the parent of $S[i+k-1]$
is within $S[i..j]$ or not (i.e., the parent is outside if $PD(S)[i+k-1] \geq k$).

\begin{equation}
    PD(S[i..j])[k] = \begin{cases}
           0 &\text{if $PD(S)[i+k-1] \geq k$} \label{eq:substr}\\
           PD(S)[i+k-1] &\text{otherwise.}  
        \end{cases} 
\end{equation}

For example, the parent-distance representation of string $S = (2, 7, 5, 6, 4, 3, 1)$ is $PD(S) = (0, 1, 2, 1, 4, 5, 0)$. For $PD(S[2..7])$, we can use the above equation and compute the value at each position in constant time, getting $PD(S[2..7]) = (0, 0, 1, 0, 0, 0)$. 

\subsection{Failure function}
We can define a failure function similar to the one used in the KMP algorithm \cite{KMP}.

\begin{definition}
    \emph{(Failure function) The} failure function \emph{$\pi$ of string $P$ is an integer string such that:}
    \begin{equation*}
        \pi[q] = \begin{cases}
           \max \{ k : CT(P[1..k]) = CT(P[q-k+1..q])   \text{\emph{ for }} 1 \leq k < q \} & \text{\emph{if $q > 1$}}\\
            0 &\text{\emph{if $q = 1$}}
        \end{cases}
    \end{equation*}
\end{definition}
That is, $\pi[q]$ is the largest $k$ such that the prefix and the suffix of $P[1..q]$ of length $k$ have the same Cartesian trees. For example, assuming that $P = (5, 7, 4, 6, 1, 3, 2)$, the corresponding failure function is $\pi = (0, 1, 1, 2, 3, 4, 1)$. We can see that $CT(P[1..4]) = CT(P[3..6])$ from $\pi[6] = 4$.
We will present an algorithm to compute the failure function of a given string in Section 3.5.

\begin{algorithm}[t]
\caption{Text search of Cartesian tree matching}\label{alg:textsearch}
\begin{algorithmic}[1]
\Procedure{CARTESIAN-TREE-MATCH}{$T[1..n], P[1..m]$}
    \State $PD(P)$ $\gets$ PARENT-DIST-REP($P$)
    \State $\pi$ $\gets$ FAILURE-FUNC($P$)
    \State $len \gets 0$
    
    \State $DQ \gets$ an empty deque
    \For{$i \gets 1$ \textbf{to} $n$}
        \State Pop elements $(value,index)$ from back of $DQ$ such that $value > T[i]$ 
            
        \While{$len \neq 0$}
            \If{$PD(T[i-len..i])[len+1] = PD(P)[len+1]$}
                \State \textbf{break}
            \Else
                \State $len \gets \pi[len]$
                
                \State Delete elements $(value,index)$ from front of $DQ$ such that $index < i-len$
            \EndIf
        \EndWhile
        
        \State $len \gets len + 1$
        \State $DQ.push\_back((T[i], i))$
        
        \If{$len = m$}
            \State \textbf{print} ``Match occurred at $i-m+1$''
            \State $len \gets \pi[len]$
	\State Delete elements $(value,index)$ from front of $DQ$ such that $index \leq i-len$	
        \EndIf
    \EndFor
\EndProcedure
\end{algorithmic}
\end{algorithm}

\subsection{Text search}
As in the original KMP text search algorithm, we can use the failure function in order to achieve linear time text search: scan the text from left to right, and use the failure function every time we find a mismatch between the text and the pattern. We apply this idea to Cartesian tree matching.

In order to perform a text search using $O(m)$ space, we compute the parent-distance representation of the text \emph{online} as we read the text, so that we don't need to store the parent-distance representation of the whole text, which would cost $O(n)$ space. Furthermore, among the text characters which are matched with the pattern, we only have to store elements that form a non-decreasing subsequence by using a \emph{deque} (instead of a stack in Section 3.2) in order to delete elements in front. Using this idea, we can keep the size of the deque to be always smaller than or equal to $m$. Therefore, we can perform the text search using $O(m)$ space. 
Algorithm \ref{alg:textsearch} shows the text search algorithm of Cartesian tree matching.
In line 9 we need to compute $x=PD(T[i-len..i])[len+1]$. If the deque is empty, then $x=0$. Otherwise, let $(value,index)$ be the element at the back of the deque. Then $x=i-index$. This computation takes constant time.
Just before line 14, we do not compare $PD(T[i])$ and $PD(P)[1]$ when $len=0$, because they always match. Therefore, we can safely perform line 14.

\subsection{Computing failure function}
We compute the failure function $\pi$ in a way similar to the text search, as in the KMP algorithm. However, we can compute the parent-distance representation of the pattern in $O(m)$ time before we compute the failure function. Hence we don’t need a deque and the computation is slightly simpler than text search. Algorithm \ref{alg:failurefunction} shows the procedure to compute the failure function.

\begin{algorithm}[t]
\caption{Computing failure function in Cartesian tree matching}\label{alg:failurefunction}
\begin{algorithmic}[1]
\Procedure{FAILURE-FUNC}{$P[1..m]$} 
    \State $PD(P)$ $\gets$ PARENT-DIST-REP($P$)
    \State $len \gets 0$
    \State $\pi[1] \gets 0$
    \For{$i \gets 2$ \textbf{to} $m$}
        \While{$len \neq 0$}
            \If{$PD(P[i-len..i])[len+1] = PD(P[1..len+1])[len+1]$}
                \State \textbf{break}
            \Else
                \State $len \gets \pi[len]$
            \EndIf
        \EndWhile
        
        \State $len \gets len + 1$
        \State $\pi[i] \gets len$
    \EndFor
\EndProcedure
\end{algorithmic}
\end{algorithm}

\subsection{Correctness and time complexity}
Since our algorithm for Cartesian tree matching including text search and the computation of the failure function follow the KMP algorithm, it is easy to see that our algorithm correctly finds all occurrences (in the sense of Cartesian tree matching) of the pattern in the text.
Since our algorithm checks one character of the parent-distance representation in constant time, it takes $O(n)$ time for text search and $O(m)$ time to compute the failure function, as in KMP algorithm. Therefore, our algorithm requires $O(m+n)$ time for Cartesian tree matching using $O(m)$ space.

\subsection{Cartesian tree signature}

There is an alternative representation of Cartesian trees, called \emph{Cartesian tree signature} \cite{CartesianTree:BitRepresentation}. The Cartesian tree signature of $S[1..n]$ is an array $L[1..n]$ such that $L[i]$ equals the number of the elements popped from the stack in the $i$-th iteration of Algorithm \ref{alg:computepdref}. Furthermore, the Cartesian tree signature can be represented as a bit string $1^{L[1]}01^{L[2]}0\cdots 1^{L[n]}0$ of length less than $2n$, which is a succinct representation of a Cartesian tree. For example, the Cartesian tree signature of string $S = (2, 7, 5, 6, 4, 3, 1)$ is $L = (0, 0, 1, 0, 2, 1, 2)$, and its corresponding bit string is $0010011010110$. 
 
We can use this representation to perform Cartesian tree matching. While we compute the Cartesian tree signature, we store one more array $D[1..n]$, which is defined as follows: If $S[i]$ is never popped out from the stack, $D[i] = 0$. Otherwise, let $S[j]$ be the value which popped $S[i]$ out from the stack, and $D[i] = j - i$. For string $S = (2, 7, 5, 6, 4, 3, 1)$, we have $D = (6, 1, 2, 1, 1, 1, 0)$. 

Using array $D$, we can delete one character at the front of string $S[1..n]$ in constant time. In order to get Cartesian tree signature $L'$ and its corresponding $D'$ for $S[2..n]$, we do the following: If $D[1] > 0$, we decrease $L[D[1] + 1]$ by one and erase $L[1]$ from $L$. If $D[1] = 0$, we just erase $L[1]$. After that, we delete $D[1]$ from $D$ to get $D'$. For example, if we want to delete one character at the front of $S = (2, 7, 5, 6, 4, 3, 1)$, we decrease $L[D[1] + 1] = L[7]$ by one, and delete $L[1]$ and $D[1]$. This results in $L' = (0, 1, 0, 2, 1, 1)$ and $D' = (1, 2, 1, 1, 1, 0)$. These arrays are the correct Cartesian tree signature and its corresponding array $D$ of $S[2..7] = (7, 5, 6, 4, 3, 1)$. In this way, we can perform Algorithm \ref{alg:textsearch} using the Cartesian tree signature. Computing the failure function can also be done in a similar way.

Note that the Cartesian tree signature can represent a Cartesian tree using less space than the parent-distance representation, but it needs an auxiliary array $D$ to perform string matching, which uses the same space as the parent-distance representation. For Cartesian tree matching, therefore, it uses more space than Algorithm \ref{alg:textsearch}.

\section{Multiple Pattern Matching in \texorpdfstring{$O((n+m) \log k)$ Time}{Lg}}
In this section we extend Cartesian tree matching to the case of multiple patterns.  Definition \ref{def:multiplepatternmatch} gives the formal definition of multiple pattern matching.
\begin{definition}
\label{def:multiplepatternmatch}
    \emph{(Multiple pattern Cartesian tree matching) Given a text $T[1..n]$ and patterns $P_1[1..m_1], P_2[1..m_2], ..., P_k[1..m_k]$, where $m = m_1 + m_2 + \cdots + m_k$, } multiple pattern Cartesian tree matching \emph{is to find every position in the text which matches at least one pattern, i.e., it has the same Cartesian tree as that of at least one pattern.}
\end{definition}
We modify the Aho-Corasick algorithm \cite{Aho-Corasick} using the parent-distance representation defined in Section 3.1 to do multiple pattern matching in $O((n+m) \log k)$ time.

\begin{figure}
    \centering
    \includegraphics[height=6.5cm]{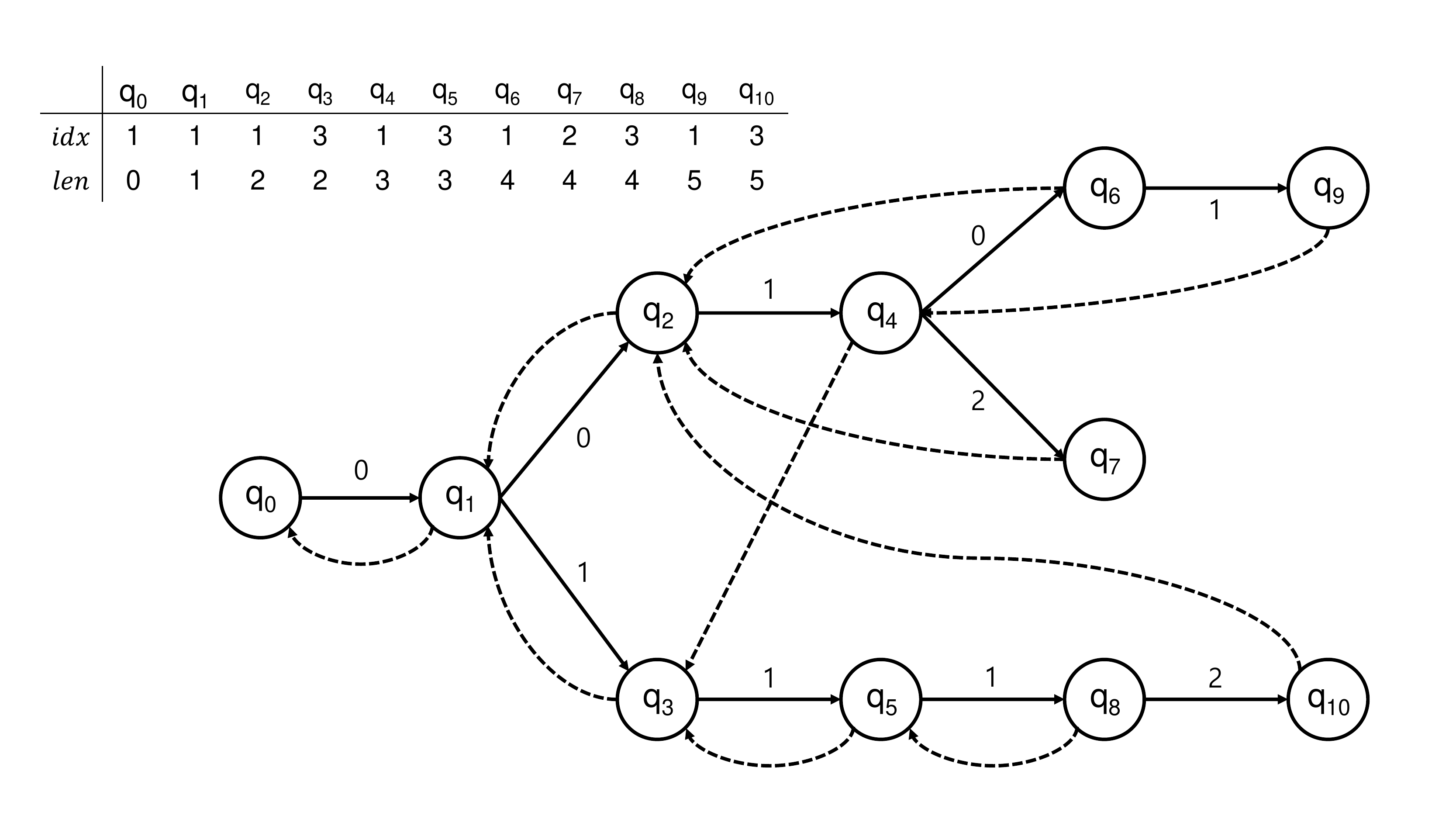}
    \caption{Aho-Corasick automaton for $P_1 = (4, 2, 3, 1, 5), P_2 = (3, 1, 4, 2), P_3 = (1, 2, 3, 5, 4)$}
    \label{fig:ahocorasickexample}
\end{figure}

\subsection{Constructing the Aho-Corasick automaton} 

Instead of using the patterns themselves in the Aho-Corasick automaton, we use their parent-distance representations to make an automaton. Each node in the automaton corresponds to the prefix of the parent-distance representation of some pattern. 
We maintain two integers $idx$ and $len$ for every node such that the node corresponds to the parent-distance representation of the pattern prefix $P_{idx}[1..len]$. If there are more than one possible indexes, we store the smallest one.
Each node also has a state transition function $trans(x)$, which gets an integer $x$ as an input and returns the next node, or report that there is no such node. We can construct the trie and the state transition function for every node in $O(m\log k)$ time, assuming that we use a balanced binary search tree to implement the transition function. Figure \ref{fig:ahocorasickexample} shows an Aho-Corasick automaton for three patterns $P_1 = (4, 2, 3, 1, 5), P_2 = (3, 1, 4, 2), P_3 = (1, 2, 3, 5, 4)$, where we use the parent-distance representations of the patterns, $PD(P_1) = (0, 0, 1, 0, 1), PD(P_2) = (0, 0, 1, 2), PD(P_3) = (0, 1, 1, 1, 2)$ to construct the automaton. 

\begin{algorithm}[t]
\caption{Computing failure function in multiple pattern matching}\label{alg:multiplefailurefunction}
\begin{algorithmic}[1]
\Procedure{MULTIPLE-FAILURE-FUNC}{$P_1$, $P_2$, ..., $P_k$}
    \For{$i \gets 1$ \textbf{to} $k$}
        \State $PD(P_i) \gets$ PARENT-DIST-REP($P_i$)
    \EndFor
     
    \State $TR \gets$ Build trie with $PD(P_i)$'s
    
    \For{$node \gets$ breadth-first traversal of the trie}
        \State $len \gets len[node]$
        \State $idx \gets idx[node]$
        
        \State $\pi[node] \gets TR.root$
        \State $ptr \gets$ parent of $node$ in the trie
        \While{$ptr \neq TR.root$}
            \State $ptr \gets \pi[ptr]$
            \State $plen \gets len[ptr]$
            \State $x \gets PD(P_{idx}[len-plen..len])[plen+1]$
            
            \If {$ptr.trans(x)$ \textbf{exists}}
                \State $\pi[node] \gets ptr.trans(x)$
                \State \textbf{break}
            \EndIf
        \EndWhile
    \EndFor
\EndProcedure
\end{algorithmic}
\end{algorithm}

The failure function $\pi$ of the Aho-Corasick automaton is defined as follows: Let $q_i$ be a node in the automaton, and $s_i$ be the substring that node $q_i$ represents in the trie. Let $s_j$ be the longest proper suffix of $s_i$ which matches (in the sense of Cartesian tree matching) prefix $s_k$ of some pattern $P_k$. The failure function of $q_i$ is defined as node $q_k$ (i.e., $\pi[q_i]=q_k$). The dotted lines in Figure \ref{fig:ahocorasickexample} shows the failure function of each node. For example, node $q_7$ represents $P_2[1..4]$, and its failure function $q_2$ represents $P_2[1..2]$. We can see that $P_2[1..2]$ matches $P_2[3..4]$ (i.e., $PD(P_2[1..2])=PD(P_2[3..4])=(0,0)$), which is the longest proper suffix of $P_2[1..4]$ that matches a  prefix of some pattern. Note that the parent-distance representation of $s_k$ may not be the suffix of the parent-distance representation of $s_i$. For example, $q_7$ has the parent-distance representation $(0, 0, 1, 2)$, but its failure function $q_2$ has the parent-distance representation $(0, 0)$ which is not a suffix of $(0, 0, 1, 2)$.

Algorithm \ref{alg:multiplefailurefunction} computes the failure function of the trie. As in the original Aho-Corasick algorithm, we traverse the trie with breadth-first order (except the root) and compute the failure function. The main difference between Algorithm \ref{alg:multiplefailurefunction} and the Aho-Corasick algorithm is at line 13, where we decide the next character to match. 
According to the definition of the trie, $node$ corresponds to the parent-distance representation of $P_{idx}[1..len]$, and so the parent of $node$ corresponds to the parent-distance representation of $P_{idx}[1..len-1]$. 
In the while loop from line 10 to 16, $ptr$ corresponds to the parent-distance representation of some suffix of $P_{idx}[1..len-1]$, because $ptr$ is a node that can be reached from the parent of $node$ following the failure links. Since $ptr$ corresponds to some string of length $plen$, we can conclude that $ptr$ represents $P_{idx}[len-plen..len-1]$. We want to check whether $P_{idx}[len-plen..len]$ matches some node in the trie, so we should check whether $ptr$ has the transition using $x = PD(P_{idx}[len-plen..len])[plen+1]$. If $ptr$ has the transition $ptr.trans(x)$, it corresponds to $P_{idx}[len-plen..len]$, and we can conclude that $\pi[node] = ptr.trans(x)$. If $ptr$ doesn't have such a transition, there is no node that represents $P_{idx}[len-plen..len]$, and thus we have to continue the loop.

For example, suppose that we compute the failure function of $q_7$ in Figure \ref{fig:ahocorasickexample}. From $idx[q_7] = 2$ and $len[q_7] = 4$, we know that $q_7$ represents $P_2[1..4]$, and so $q_4$, which is the parent of $q_7$, represents $P_2[1..3]$. We begin the while loop starting from $ptr = \pi[q_4] = q_3$. Since $len[q_3] = 2$, we know that $q_3$, which represents $P_3[1..2]$, matches $P_2[2..3]$. In order to check whether $P_2[2..4]$ matches some node in the trie, we compute $x = PD(P_2[2..4])[3] = 2$ and check whether $q_3.trans(x)$ exists. Since there is no such transition, we continue the while loop with $ptr = \pi[q_3] = q_1$. We know that $q_1$, which represents $P_1[1..1]$, matches $P_2[3..3]$ from $len[q_1] = 1$. In order to check whether $P_2[3..4]$ matches some node, we compute $x = PD(P_2[3..4])[2] = 0$ and check whether $q_1.trans(x)$ exists. Since there is such a transition, we conclude that $\pi[q_7] = q_1.trans(0) = q_2$. Note that $x$ may change during the while loop, which is not the case in the Aho-Corasick algorithm.

While computing the failure function, we can also compute the output function in the same way as the Aho-Corasick algorithm. The output function of node $q_i$ is the set of patterns which match some suffix of $s_i$. This function is used to output all possible matches at the node. 

\subsection{Multiple pattern matching}
Using the automaton defined above, we can solve multiple pattern Cartesian tree matching in $O(n \log k)$ time. The text search algorithm is essentially the same as that of the Aho-Corasick algorithm, following the trie and using the failure links in case of any mismatches. As in the single pattern case, we compute the parent-distance representation of the text online in the same way as Algorithm \ref{alg:textsearch} (using a deque) to ensure $O(m)$ space. 
The time complexity of our multiple pattern Cartesian tree matching is $O((n+m) \log k)$ using $O(m)$ space, where the $\log k$ factor is included due to the binary search tree in each node. Since there can be at most $k$ outgoing edges from each node, we can perform an operation in the binary search tree in $O(\log k)$ time. Combined with the time-complexity analysis of the Aho-Corasick algorithm, this shows that our algorithm has the time complexity of $O((n+m) \log k)$. We can reduce the time complexity further to randomized $O(n+m)$ time by using a hash instead of a binary search tree \cite{SuffixTree:Cole}.

\section{Cartesian Suffix Tree in Randomized \texorpdfstring{$O(n)$}{Lg} Time}
In this section we apply the notion of Cartesian tree matching to the suffix tree as in the cases of parameterized matching and order-preserving matching \cite{SuffixTree:Parameterized, Order-preservingIndexing}. We first define the Cartesian suffix tree, and show that it can be built in randomized $O(n)$ time or worst-case $O(n \log n)$ time using the result from Cole and Hariharan \cite{SuffixTree:Cole}.

\subsection{Defining Cartesian suffix tree}
The Cartesian suffix tree is an index data structure that allows us to find an occurrence of a given pattern $P[1..m]$ in randomized $O(m)$ time or worst-case $O(m \log n)$ time, where $n$ is the length of the text string. In order to store the information of Cartesian suffix trees efficiently, we again use the parent-distance representation from Section 3.1. Definition \ref{def:Cartesiansuffixtree} gives the formal definition of the Cartesian suffix tree.

\begin{definition}
\label{def:Cartesiansuffixtree}
\emph{(Cartesian suffix tree) Given a string $T[1..n]$, the} Cartesian suffix tree \emph{of $T$ is a compacted trie built with $PD(T[i..n]) \cdot (-1)$ for every $1 \leq i \leq n$ (where the special character $-1$ is concatenated to the end of $PD(T[i..n])$) and string $(-1)$.}
\end{definition}

Note that we append a special character $-1$ to the end of each parent-distance representation to ensure that no string is a prefix of another string.

\begin{figure}[t]
\centering
    \includegraphics[height=8.0cm]{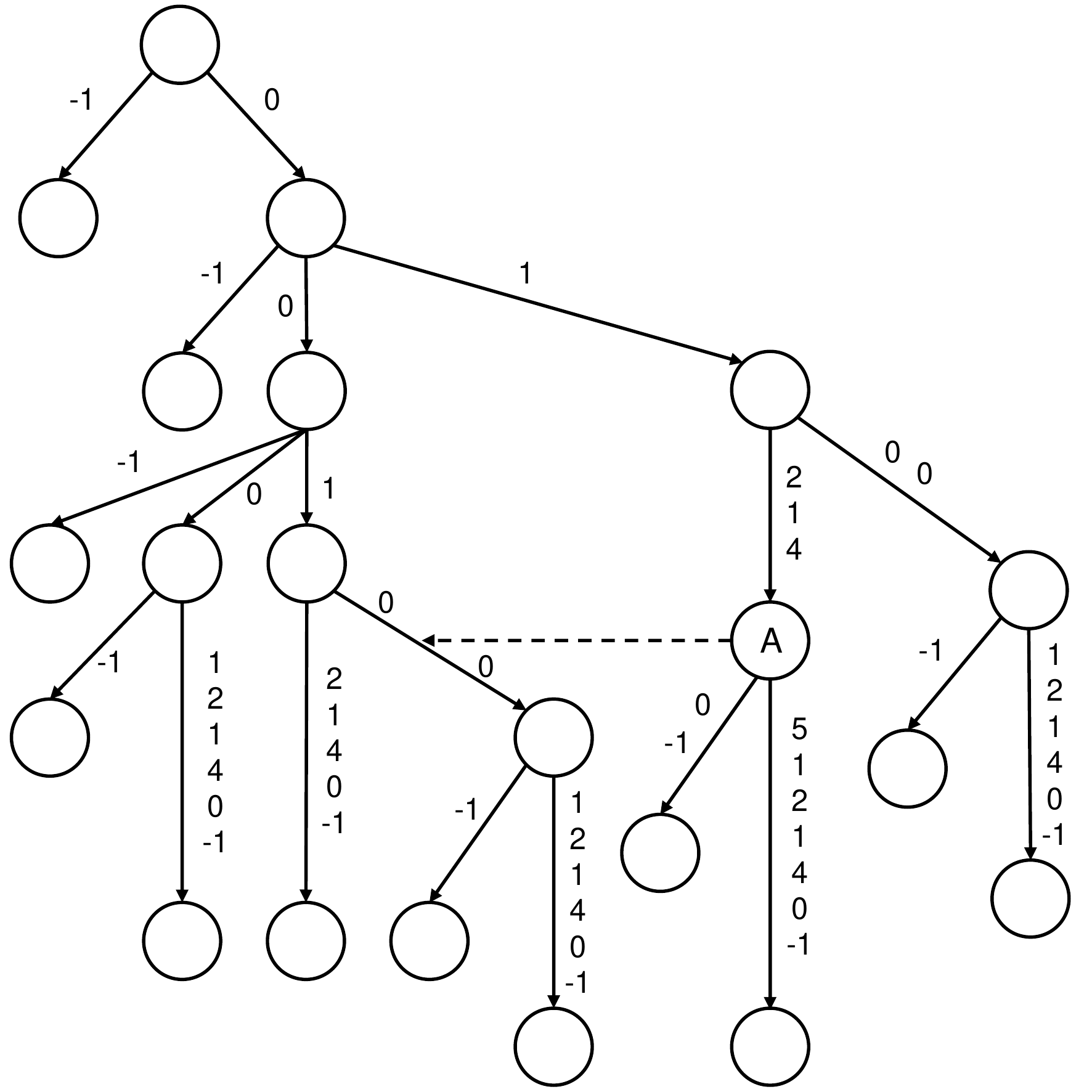}
    \caption{Cartesian suffix tree of $S = (2, 7, 5, 6, 4, 3, 11, 9, 10, 8, 1)$}
    \label{fig:suffixtreeexample}
\end{figure}

Figure \ref{fig:suffixtreeexample} shows an example Cartesian suffix tree of $T = (2, 7, 5, 6, 4, 3, 11, 9, 10, 8, 1)$. Each edge actually stores the suffix number, start position, and end position instead of the parent-distance representation itself. For example, node $A$ corresponds to substring $T[1..5]$ or $T[6..10]$, whose parent-distance representation is $PD(T[1..5])=PD(T[6..10])=(0, 1, 2, 1, 4)$. Hence, the edge that goes into node $A$ stores the suffix number $1$ or $6$, start and end positions $3$ and $5$. 

\subsection{Constructing Cartesian suffix tree}
There are several algorithms efficiently constructing the suffix tree, such as McCreight's algorithm \cite{SuffixTree:McCreight} and Ukkonen's algorithm \cite{SuffixTree:Ukkonen}. However, the \emph{distinct right context property} \cite{SuffixTree:Giancarlo, SuffixTree:Parameterized} should hold in order to apply these algorithms, which means that the suffix link of every internal node should point to an explicit node. The Cartesian suffix tree does not have the distinct right context property. In Figure \ref{fig:suffixtreeexample}, the internal node marked with $A$ does not satisfy this property because $PD(T[2..6])=PD(T[7..11])=(0, 0, 1, 0, 0)$ and thus there is no explicit node corresponding to parent-distance representation $(0,0,1,0)$.

In order to handle this issue, we use an algorithm due to Cole and Hariharan \cite{SuffixTree:Cole}. This algorithm can construct a compacted trie for a \emph{quasi-suffix collection}, which satisfies the following properties:

\begin{enumerate}
    \item A quasi-suffix collection is a set of $n$ strings $s_1, s_2, ..., s_n$, where the length of $s_i$ is $n+1-i$.
    \item For any two different strings $s_i$ and $s_j$, $s_i$ should not be a prefix of $s_j$.
    \item For any $i$ and $j$, if $s_i$ and $s_j$ have a common prefix of length $l$, $s_{i+1}$ and $s_{j+1}$ should have a common prefix of length at least $l-1$.
\end{enumerate}

A collection of parent-distance representations for the Cartesian suffix tree satisfies all of the above properties. The first two properties are trivial. Furthermore, if $s_i = PD(T[i..n]) \cdot (-1)$ and $s_j = PD(T[j..n]) \cdot (-1)$ have a common prefix of length $l$, i.e., $PD(T[i..i+l-1]) = PD(T[j..j+l-1])$, we can show that $PD(T[i+1..i+l-1]) = PD(T[j+1..j+l-1])$ by Equation~\ref{eq:substr}. Therefore, $s_{i+1} = PD(T[i+1..n]) \cdot (-1)$ and $s_{j+1} = PD(T[j+1..n]) \cdot (-1)$ have a common prefix of length $l-1$ or more, showing the third property holds.

One more property we need to perform Cole and Hariharan's algorithm is a \emph{character oracle}, which returns the $i$-th character of $s_j$ in constant time. We can do this in constant time using Equation~\ref{eq:substr}, once the parent-distance representation of $T$ is computed. 

Since we have all properties needed to perform Cole and Hariharan's algorithm, we can construct a Cartesian suffix tree in randomized $O(n)$ time using $O(n)$ space \cite{SuffixTree:Cole}.
In the worst case, it can be built in $O(n \log n)$ time by using a binary search tree instead of a hash table to store the children of each node in the suffix tree, because the alphabet size $|\Sigma|$ is $O(n)$. We can also modify our algorithm to construct a Cartesian suffix tree online, using the idea in \cite{SuffixTree:Parameterized2, SuffixTree:2DOnline}.

\section{Conclusion}

We have defined Cartesian tree matching and the parent-distance representation of a Cartesian tree. We developed a linear time algorithm for single pattern matching and an $O((n+m) \log k)$ deterministic time or $O(n+m)$ randomized time algorithm for multiple pattern matching. Finally, we defined an index data structure called Cartesian suffix tree, and showed that it can be constructed in $O(n)$ randomized time. We believe that the notion of Cartesian tree matching, which is a new metric on string matching and indexing over numeric strings, can be used in many applications.

There have been many works on approximate generalized matching. For example, there are results for approximate order-preserving matching \cite{Order-preservingMatching:Approximate}, approximate jumble matching \cite{JumbledMatching:Approximate}, approximate swapped matching \cite{SwappedMatching:Approximate}, and approximate parameterized matching \cite{ParameterizedMatching:Approximate2, ParameterizedMatching:Approximate}. There are also results on computing the period of a generalized string, such as computing the period in the order-preserving model \cite{Order-preservingMatching:Period}. Since Cartesian tree matching is first introduced in this paper, many problems including approximate matching and computing the period in the Cartesian tree matching model are future research topics.

\section*{Acknowledgments}
S.G. Park and K. Park were supported by Institute for Information \& communications Technology Promotion(IITP) grant funded by the Korea government(MSIT) (No. 2018-0-00551, Framework of Practical Algorithms for NP-hard Graph Problems).
A. Amir and G.M. Landau were partially supported by the Israel Science Foundation grant 571/14, and Grant No. 2014028 from the United States-Israel Binational Science Foundation (BSF).

\bibliography{bibliography}

\end{document}